\documentclass[sigconf]{aamas}  

\usepackage{booktabs}
\usepackage{framed}

\setcopyright{ifaamas}  
\acmDOI{}  
\acmISBN{}  
\acmConference[AAMAS'18]{Proc.\@ of the 17th International Conference on Autonomous Agents and Multiagent Systems (AAMAS 2018)}{July 10--15, 2018}{Stockholm, Sweden}{M.~Dastani, G.~Sukthankar, E.~Andr\'{e}, S.~Koenig (eds.)}  
\acmYear{2018}  
\copyrightyear{2018}  
\acmPrice{}  

\begin{document}

\title{Selling Multiple Items via Social Networks}  





\author{Dengji Zhao$^\ast$, Bin Li$^\dagger$, Junping Xu$^\ast$, Dong Hao$^\dagger$, and Nicholas R. Jennings$^\ddagger$}
\affiliation{%
  \institution{$^{\ast}$ShanghaiTech University, Shanghai, China}
}
\email{{zhaodj, xujp}@shanghaitech.edu.cn}
\affiliation{\institution{$^{\dagger}$University of Electronic Science and Technology of China, Chengdu, China}
}
\email{{libin@std., haodong@}uestc.edu.cn}
\affiliation{%
  \institution{$^{\ddagger}$Imperial College London, London, United Kingdom}
}
\email{n.jennings@imperial.ac.uk}

\begin{abstract}  
We consider a market where a seller sells multiple units of a commodity in a social network. Each node/buyer in the social network can only directly communicate with her neighbours, i.e. the seller can only sell the commodity to her neighbours if she could not find a way to inform other buyers. In this paper, we design a novel promotion mechanism that incentivizes all buyers, who are aware of the sale, to invite all their neighbours to join the sale, even though there is no guarantee that their efforts will be paid. While traditional sale promotions such as sponsored search auctions cannot guarantee a positive return for the advertiser (the seller), our mechanism guarantees that the seller's revenue is better than not using the advertising. More importantly, the seller does not need to pay if the advertising is not beneficial to her. 
\end{abstract}

%


\maketitle



\section{Introduction}
Marketing is one of the key operations for a service or product to survive. To do that, companies often use newspapers, tv, social media, search engines to do advertisements. Indeed, most of the revenue of social media and search engines comes from paid advertisements. According to Statista, Google's ad revenue amounted to almost $79.4$ billion US dollars in 2016. However, whether all the advertisers actually benefit from their advertisements is not clear and is difficult to monitor. Although most search engines use market mechanims like generalised second price auctions to allocate advertisements and only charge the advertisers when users click their ads, not all clicks lead to a purchase~\cite{gsp2007,varian2009online}. That said, the advertisers may pay user clicks that have no value to them.  

Therefore, in this paper, we propose a novel advertising mechanism for a seller (to sell services or products) that does not charge the seller unless the advertising brings an increase in revenue. We model all (potential) buyers of a service/product as a large social network where each buyer is linked with some other buyers (known as neighbours). The seller is also located somewhere in the social network. Before the seller finds a way to inform more buyers about her sale, she can only sell her products to her neighbours. In order to attract more buyers to increase her revenue, the seller may pay to advertise the sale via newspapers, social media, search engines etc. to reach/inform more potential buyers in the social network. However, if the advertisements do not bring any valuable buyers, the seller loses the investment on the advertisements.

Our advertising mechanism does not rely on any third party such as newspapers or search engines to do the advertisements. The mechanism is owned by the seller. The seller just needs to invite all her neighbours to join the sale, then her neighbours will further invite their neighbours and so on. In the end, all buyers in the social network will be invited to participate in the sale. Moreover, all buyers are not paid in advance for their invitations and they may not get paid if their invitations are not beneficial to the seller. Although some buyers may never get paid for their efforts in the advertising, they are still incentivized to do so, which is one of the key features of our advertising mechanism. This significantly differs from existing advertising mechanisms used on the Internet.

More importantly, our advertising mechanism not only incentivizes all buyers to do the advertising, but also guarantees that the seller's revenue increases. That is, her revenue is never worse than the revenue she can get if she only sells the items to her neighbours.

A special case of this problem was investigated by Li et al.~\cite{li2017mechanism}. They have considered the setting when the seller sells only one item and proposed a mechanism, called the information diffusion mechanism, that guarantees that all buyers will truthfully report their willing payments (i.e. valuations) and also invite all their neighbours to join the sale. They have shown that the mechanism gives a revenue which is at least the revenue the seller can receive with a second price auction among only the seller's neighbours.

This paper generalises the mechanism proposed by Li et al.~\cite{li2017mechanism} to settings where the seller sells multiple items. This generalisation still guarantees that reporting their true valuations and inviting all their neighbours is a dominant strategy for all buyers who are aware of the sale. Moreover, the revenue of the seller is also improved compared with the revenue she can achieve with traditional market mechanisms such as VCG~\cite{vickrey1961counterspeculation, clarke1971multipart, groves1973incentives}. 

Maximising the seller's revenue has been well studied in the literature, but the existing models assumed that the buyers are all known to the seller and the aim is to maximize the revenue among the fixed number of buyers. Given the number of buyers is fixed, if we have some prior information about their valuations, Myerson~\cite{myerson1981optimal} proposed a mechanism by adding a reserve price to the original VCG mechanism. Myerson's mechanism maximises the seller's revenue, but requires the distributions of buyers' valuations to compute the reserve price. Without any prior information about the buyers' valuations, we cannot design a mechanism that can maximise the revenue in all settings (see Chapter 13 of~\cite{nisan2007algorithmic} for a detailed survey). Goldberg et al.~\cite{goldberg2001competitive,goldberg2001digital} have considered how to optimize the revenue for selling multiple homogeneous items such as digital goods like software (unlimited supply). Especially, the seller can choose to sell less with a higher price to gain more. 

In terms of incentivizing people to share information (like buyers inviting their neighbours), there also exists a growing body of work~\cite{kempe2003maximizing, rogers2010diffusion, pickard2011time, emek2011mechanisms}. Their settings are essentially different from ours however. They considered either how information is propagated in a social network or how to design reward mechanisms to incentivize people to invite more people to accomplish a challenge together. The solution offered by the MIT team under the DARPA Network Challenge is a nice example, where they designed a novel reward mechanism to share the award if they win the challenge to attract many people via social network to join the team, which eventually helped them to win~\cite{pickard2011time}. 

The remainder of the paper is organized as follows. Section~\ref{sect_model} describes the model of the advertising problem. Section~\ref{sect_idm} briefly reviews the mechanism proposed by Li et al.~\cite{li2017mechanism}. Section~\ref{sect_gidm} gives our generalisation and its key properties are analysed in Section~\ref{sect_properties}. Finally, we conclude in Section~\ref{sect_con}.

%
%
%
%

\section{The Model}
\label{sect_model}
We consider a seller $s$ sells $\mathcal{K} \geq 1$ items in a social network. In addition to the seller, the social network consists of $n$ nodes denoted by $N = \{1, \cdots, n\}$, and each node $i\in N \cup \{s\}$ has a set of neighbours denoted by $r_i \subseteq N \cup \{s\}$. Each $i\in N$ is a buyer of the $\mathcal{K}$ items.

For simplicity, we assume that the $\mathcal{K}$ items are homogeneous and each buyer $i\in N$ requires at most one unit of the item and has a valuation $v_i \geq 0$ for one or more units.

Without any advertising, seller $s$ can only sell to her neighbours $r_s$ as she is not aware of the rest of the network and the other buyers also do not know the seller $s$. In order to maximize $s$'s profit, it would be better if all buyers in the network could join the sale.

Traditionally, the seller may pay some of her neighbours to advertise the sale to their neighbours, but the neighbours may not bring any valuable buyers and cost the seller money for the advertisement. Therefore, our goal here is to design a a kind of cost-free advertising mechanism such that all buyers, who are aware of the sale, are incentivized to invite all their neighbours to join the sale with no guarantee that their efforts will be paid. Li et al.~\cite{li2017mechanism} have shown that this is achievable when $\mathcal{K}=1$. In this paper, we generalize their approach to $\mathcal{K}\geq 1$. 

Let us first formally describe the model. Let $\theta_i = (v_i, r_i)$ be the \emph{type} of buyer $i\in N$, $\theta = (\theta_1,\cdots,\theta_n)$ be the type profile of all buyers and $\theta_{-i}$ be the type profile of all buyers except $i$. $\theta$ can also be represented by $(\theta_i, \theta_{-i})$. Let $\Theta_i$ be the type space of buyer $i$ and $\Theta$ be the type profile space of all buyers.

The advertising mechanism consists of an \emph{allocation policy} $\pi$ and a \emph{payment policy} $x$. The mechanism requires each buyer, who is aware of the sale, to report her valuation to the mechanism and invite/inform all her neighbours about the sale. Let $v_i^\prime$ be the valuation report of buyer $i$ and $r_i^\prime \subseteq r_i$ be the neighbours $i$ has invited. Let $\theta_i^\prime = (v_i^\prime, r_i^\prime)$ and $\theta^\prime = (\theta_1^\prime,\cdots, \theta_n^\prime)$, where $\theta_j^\prime = nil$ if $j$ has never been invited by any of her neighbours $r_j$ or $j$ does not want to participate. Given the action profile $\theta^\prime$ of all buyers, $\pi_i(\theta^\prime) \in \{0,1\}$, $1$ means that $i$ receives one item, while $0$ means $i$ does not receive any item. $x_i(\theta^\prime) \in \mathbb{R}$ is the payment that $i$ pays to the mechanism, $x_i(\theta^\prime) < 0$ means that $i$ receives $|x_i(\theta^\prime)|$ from the mechanism. 

\begin{definition}
	Given an action profile $\theta^\prime$ of all buyers, an \emph{invitation chain} from the seller $s$ to a buyer $i$ is a buyer sequence of $(s,j_1,\cdots, j_l,j_{l+1}, \cdots, j_m,i)$ such that $j_1\in r_s$, for all $1<l \leq m$ $j_l \in r_{j_{l-1}}^\prime$, $i\in r_{j_m}^\prime$ and no buyer appears twice in the sequence, i.e. it is acyclic.
\end{definition}

A buyer cannot invite buyers who are not her neighbours and a buyer who has never been invited by her neighbours cannot join the sale, therefore not all action profiles are feasible. 

\begin{definition}
	Given the buyers' type profile $\theta$, an action profile $\theta^\prime$ is \emph{feasible} if for all $i\in N$, 
	\begin{itemize}
		\item $\theta_i^\prime \neq nil$ if and only if there exists an \emph{invitation chain} from the seller $s$ to $i$ following the action profile of $\theta_{-i}^\prime$.
		\item if $\theta_i^\prime \neq nil$, then $r_i^\prime \subseteq r_i$.
	\end{itemize}
	Let $\mathcal{F}(\theta)$ be the set of all feasible action profiles of all buyers under type profile $\theta$.
\end{definition}

The advertising mechanism $(\pi,x)$ is defined only on feasible action profiles. In the following, we define the related properties of the mechanism.

\begin{definition}
	An allocation $\pi$ is \emph{feasible} if for all $\theta\in \Theta$, for all $\theta^\prime \in \mathcal{F}(\theta)$,
	\begin{itemize}
		\item for all $i\in N$, if $\theta_i^\prime = nil$, then $\pi_i(\theta^\prime) = 0$.
		\item $\sum_{i\in N}\pi_i(\theta^\prime) \leq \mathcal{K}$.
	\end{itemize}
\end{definition}
A feasible allocation does not allocate items to buyers who have never participated and it does not allocate more than $\mathcal{K}$ items. In the rest of this paper, we only consider feasible allocations.

\begin{definition}
	An allocation $\pi$ is \emph{efficient} if for all $\theta\in \Theta$, for all $\theta^\prime \in \mathcal{F}(\theta)$,
	$$\pi \in {\arg\max}_{\pi^\prime \in \Pi} \sum_{i\in N, \theta'_i \neq nil} \pi^\prime_i(\theta')v_i^\prime$$
	where $\Pi$ is the set of all feasible allocations.
\end{definition}

Given a buyer $i$ of type $\theta_i = (v_i, r_i)$ and a feasible action profile $\theta^\prime$, the \emph{utility} of $i$ under a mechanism $(\pi, x)$ is quasilinear and defined as:
\begin{equation*}
	u_i(\theta_i, \theta^\prime, (\pi, x)) = \pi_i(\theta^\prime)v_i  - x_i(\theta^\prime).
\end{equation*}
For simplicity, we will use $u_i(\theta_i, \theta^\prime)$ to represent $u_i(\theta_i, \theta^\prime, (\pi, x))$ as $(\pi,x)$ is clear and does not change.

We say a mechanism is individually rational if for each buyer, her utility is non-negative when she truthfully reports her valuation, no matter which neighbours she invites and what the others do. That is a buyer should not lose as long as she reports her valuation truthfully, i.e. she is not forced to invite her neighbours.

\begin{definition}
	A mechanism $(\pi,x)$ is \emph{individually rational} (IR) if $u_i(\theta_i, \theta^\prime) \geq 0$ for all $\theta\in \Theta$, for all $i\in N$, for all $\theta^\prime \in \mathcal{F}(\theta)$ such that $\theta_i^\prime = (v_i, r_i^\prime)$.
\end{definition}

Different from the traditional mechanism design settings, in this model, we want to incentivize buyers to not only just report their valuations truthfully, but also invite all their neighbours to join the sale/auction (the advertising part). Therefore, we extend the definition of incentive compatibility to cover the invitation of their neighbours. Specifically, a mechanism is incentive compatible (or truthful) if for all buyers who are invited by at least one of their neighbours, reporting their valuations truthfully to the mechanism and further inviting all their neighbours to join the sale is a dominant strategy. 

\begin{definition}
	A mechanism $(\pi, x)$ is \emph{incentive compatible} (IC) if
	$u_i(\theta_i, \theta^\prime) \geq u_i(\theta_i, \theta^{\prime\prime})$ for all $\theta\in \Theta$, for all $i\in N$, for all $\theta^\prime, \theta^{\prime\prime} \in \mathcal{F}(\theta)$ such that $\theta_i^\prime = \theta_i$ {\color{red}and for all $j\neq i$, $\theta_j^{\prime\prime} = \theta_j^{\prime}$ if there exists an invitation chain from the seller $s$ to $j$ following the action profile of $(\theta_i^{\prime\prime}, \theta_{-i}^{\prime})$}.
\end{definition}

Given a feasible action profile $\theta^\prime$ and a mechanism $(\pi, x)$, the seller's \emph{revenue} generated by $(\pi,x)$ is defined by the sum of all buyers' payments, denoted by $R^{(\pi,x)}(\theta^\prime) = \sum_{i\in N} x_i(\theta^\prime)$.

\begin{definition}
	A mechanism $(\pi, x)$ is \emph{weakly budget balanced} if for all $\theta\in \Theta$, for all $\theta^\prime \in \mathcal{F}(\theta)$,  $R^{(\pi,x)}(\theta^\prime) \geq 0$.
\end{definition}

In this paper, we design a mechanism that is IC and IR for the buyers to help the seller propagate the sale information without being paid in advance.

\section{The Information Diffusion Mechanism for $\mathcal{K}=1$}
\label{sect_idm}
In this section, we review the mechanism proposed by Li et al.~\cite{li2017mechanism} for the case of $\mathcal{K}=1$. Li et al. considered an advertising mechanism design for a seller to sell a single item in a social network. The essence of their approach is that a buyer is only rewarded for advertising if her invitations increase social welfare and the reward guarantees that inviting all neighbours is a dominant strategy for all buyers.

Their \emph{information diffusion mechanism} is outlined below:
\begin{framed}
	\noindent\textbf{Information Diffusion Mechanism (IDM)}\\
	\rule{\textwidth}{0.5pt}
	\begin{enumerate}
		\item Given a feasible action profile $\theta^\prime$, identify the buyer with the highest valuation, denoted by $i^*$.
		\item Find all \emph{diffusion critical buyers} of $i^*$, denoted by $C_{i^*}$. $j\in C_{i^*}$ if and only if without $j$'s action $\theta_{j}^\prime$, there is no invitation chain from the seller $s$ to $i^*$ following $\theta_{-j}^\prime$, i.e. $i^*$ is not able to join the sale without $j$. 
		\item For any two buyers $i,j \in C_{i^*} \cup \{i^*\}$, define an order $\succ_{i^*}$ such that $i \succ_{i^*} j$ if and only if all invitation chains from $s$ to $j$ contain $i$. 
		\item For each $i\in C_{i^*} \cup \{i^*\}$, if $i$ receives the item, the payment of $i$ is the highest valuation report without $i$'s participation. Formally, let $N_{-i}$ be the set of buyers each of whom has an invitation chain from $s$ following $\theta_{-i}^\prime$, $i$'s payment to receive the item is $p_i = max_{j\in N_{-i} \wedge \theta_j^\prime \neq nil} v_j^\prime$. 
		\item The seller initially gives the item to the buyer $i$ ranked first in $C_{i^*} \cup \{i^*\}$, let $l = 1$ and repeat the following until the item is allocated.
		\begin{itemize}
			\item if $i$ is the last ranked buyer in $C_{i^*} \cup \{i^*\}$, then $i$ receives the item and her payment is $x_i(\theta^\prime) = p_i$;
			\item else if $v_i^\prime = p_{j}$, where $j$ is the $(l+1)$-th ranked buyer in $C_{i^*} \cup \{i^*\}$, then $i$ receives the item and her payment is $x_i(\theta^\prime) = p_i$;
			\item otherwise, $i$ passes the item to buyer $j$ and $i$'s payment is $x_i(\theta^\prime) = p_i - p_j$, where $j$ is the $(l+1)$-th ranked buyer in $C_{i^*} \cup \{i^*\}$. Set $i=j$ and $l= l+1$.
		\end{itemize}
		\item The payments of all the rest buyers are zero.
	\end{enumerate}
\end{framed}


\begin{figure}
	\centering
	\includegraphics[width=2.3in]{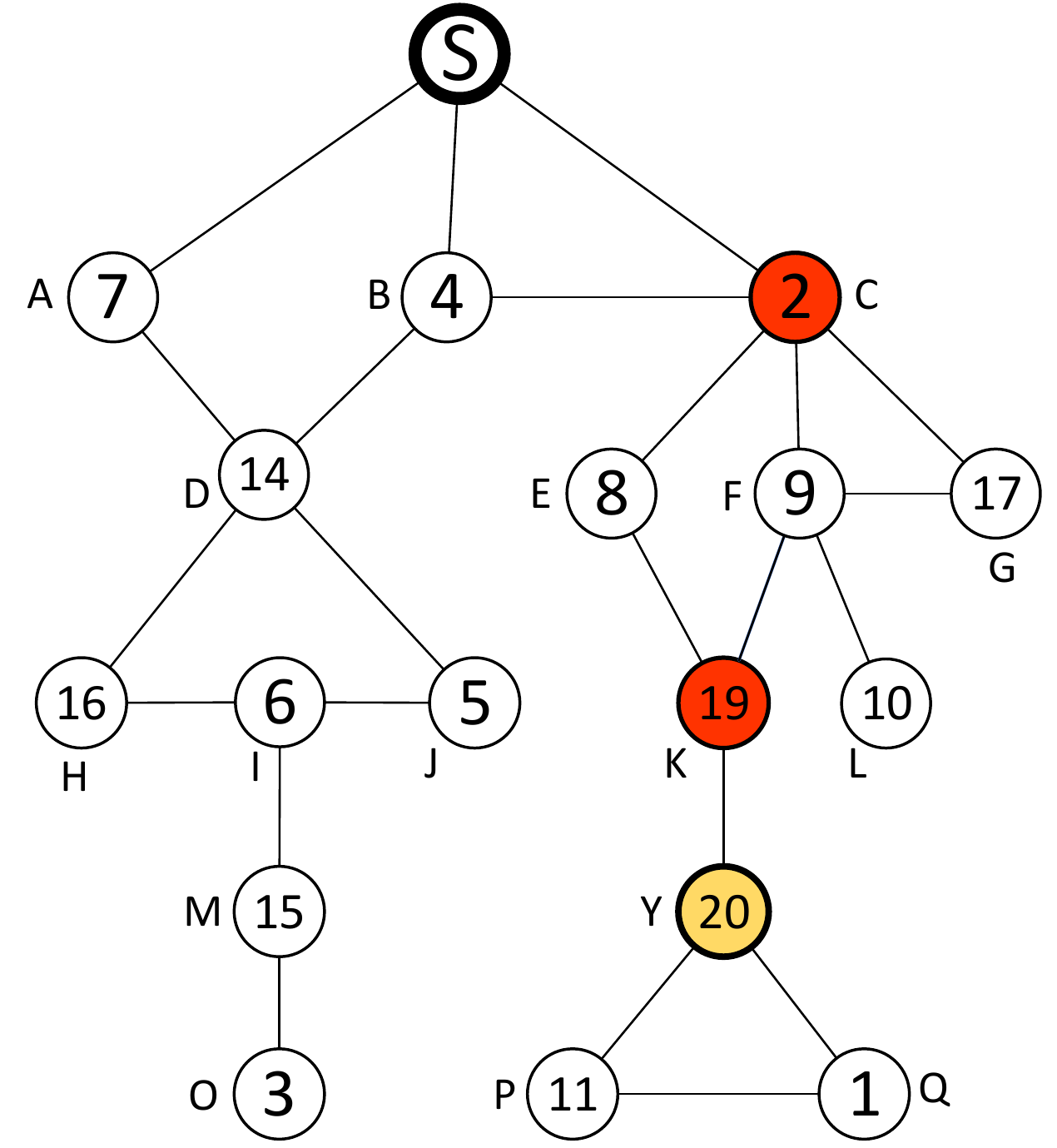}
	\caption{A running example of the information diffusion mechanism, where the seller $s$ is located at the top of the graph and is selling one item, the value in each node is the node's private valuation for receiving the item, and the lines between nodes represent neighbourhood relationship. Node $Y$ is the node with the highest valuation and $C, K$ are $Y$'s diffusion critical buyers.}\label{eg_IDM}
\end{figure}

Figure~\ref{eg_IDM} shows a social network example. 
Without any advertising, the seller can only sell the item among nodes $A$, $B$ and $C$, and her revenue cannot be more than $7$. If $A$, $B$ and $C$ invite their neighbours, these neighbours further invite their neighbours and so on, then all nodes in the social network will be able to join the sale and the seller may receive a revenue as high as the highest valuation of the social network which is $20$.

Let us run IDM on the social network given in Figure~\ref{eg_IDM}. Assume that all buyers report their valuations truthfully and invite all their neighbours, IDM runs as follows:
\begin{itemize}
	\item Step $(1)$ identifies that the buyer with the highest valuation is $Y$, i.e. $i^* = Y$. 
	\item Step $(2)$ computes $C_{i^*} = \{C, K\}$. 
	\item Step $(3)$ gives the order of $C_{i^*}\cup \{i^*\}$ as $C \succ_{i^*} K \succ_{i^*} i^*$. 
	\item Step $(4)$ defines the payments $p_i$ for all nodes in $C_{i^*}\cup \{i^*\}$, which are $p_C = 16$, $p_K = 17$ and $p_Y = 19$, the highest valuation without $C$, $K$ and $Y$'s participation respectively.
	\item Step $(5)$ first gives the item to node $C$; $C$ is not the last ranked buyer in $C_{i^*}\cup \{i^*\}$ and $v_C \neq p_K$, so $C$ passes the item to $K$ and her payment is $p_C - p_K = -1$; $K$ is not the last ranked buyer, but $v_K = p_Y$, therefore $K$ receives the item and pays $p_K$.
	\item All the rest of the buyers, including $Y$, pay nothing.
\end{itemize}

In the above example, IDM allocates the item to node $K$ and $K$ pays $17$, but $s$ does not receive all the payment, and she pays $C$ an amount of $1$ for the advertising. Therefore, the seller receives a revenue of $16$ from IDM, which is more than two times the revenue she can get without any advertising. Note that only buyer $C$ is rewarded for the information propagation as the other buyers are not critical for inviting $K$. Li et al. showed that IDM has the following desirable properties.

\begin{theorem}[Li et al.~\cite{li2017mechanism}]
IDM is incentive compatible (i.e. reporting valuations truthfully and inviting all neighbours is a dominant strategy) and individually rational. The revenue of the seller given by IDM is at least the revenue given by VCG with the seller's neighbours only.
\end{theorem}

\section{Generalised IDM for Selling Multiple Items}
\label{sect_gidm}
In this section, we present our generalisation of IDM for a seller to sell multiple homogeneous items in a social network. We assume that each buyer requires at most one item. Clearly, we cannot simply run IDM multiple times to solve the problem, as buyers who receive the item earlier would pay more, which is not incentive compatible.

Before we introduce the mechanism, we need some additional concepts.


\begin{definition}
	\label{diffusion critical node}
	Given a feasible action profile $\theta^\prime$ of all buyers, for any two buyers $i\neq j\in N$ such that $\theta_i^\prime, \theta_j^\prime \neq nil$, we say $i$ is $j$'s \emph{critical parent} if without $i$'s participation, there exists no invitation chain from the seller to $j$. If $i$ is $j$'s \emph{critical parent}, then $j$ is $i$'s \emph{critical child}. Let $\mathcal{P}_i(\theta^\prime)$ be the set of all critical parents of $i$ and $\mathcal{C}_i(\theta^\prime)$ be the set of all critical children of $i$ under action profile $\theta^\prime$.
\end{definition}

If a buyer $j \in \mathcal{P}_i(\theta^\prime)$ does not invite any of her neighbours $r_j$, $i$ cannot join the sale, while if $i$ does not invite any of her neighbours $r_i$, $\mathcal{C}_i(\theta^\prime)$ cannot join the sale.

\begin{definition}
	Given a feasible action profile $\theta^\prime \in \mathcal{F}(\theta)$, we define a \emph{partial order $\succ_{\theta^\prime}$} on all $i,j\in N$ such that $i \succ_{\theta^\prime} j$ if and only if $i\in \mathcal{P}_{j}(\theta^\prime)$.
\end{definition}
It is clear that if $i \succ_{\theta^\prime} j$ and $j \succ_{\theta^\prime} l$, then $i \succ_{\theta^\prime} l$. 

Take the example given in Figure~\ref{eg_IDM}, if all buyers act truthfully, i.e. $\theta^\prime = \theta$, for buyer $M$, we have $\mathcal{P}_M(\theta) = \{D,I\}$, $\mathcal{C}_M(\theta)=\{O\}$ and $D \succ_{\theta} I$.
Furthermore, it is easy to check that $\mathcal{P}_A(\theta) = \emptyset$ and $\mathcal{C}_A(\theta)=\emptyset$, as $A$ is directly linked to $s$ and without $A$'s participation, all other buyers can still receive invitations.

Given a feasible action profile $\theta^\prime$, we define an optimal allocation tree based on the efficient allocation and their critical parents.

\begin{definition}
	Given $\theta^\prime \in \mathcal{F}(\theta)$, an \emph{optimal allocation tree of $\theta^\prime$}, denoted by $T^{opt}(\theta^\prime)$, is defined as:
	\begin{itemize}
		\item $T^{opt}(\theta^\prime)$ is rooted at $s$,
		\item all buyers in $\{i\in N| \pi_i^{eff}(\theta^\prime) = 1\}$, denoted by $N^{opt}$, and all their critical parents $\mathcal{P}^{opt} = \cup_{i\in N^{opt}} \mathcal{P}_{i}(\theta^\prime)$ are the nodes of $T^{opt}(\theta^\prime)$, where $\pi^{eff}$ is an efficient allocation,
		\item the path from $s$ to each node $i\in N^{opt} \cup \mathcal{P}^{opt}$ is given by $\mathcal{P}_i(\theta^\prime)$ and the order $\succ_{\theta^\prime}$ in the form of $(s, j_1,\cdots, j_l, j_{l+1}, \cdots, j_m, i)$ where all $j_l, j_{l+1}\in \mathcal{P}_i(\theta^\prime)$ and $j_l \succ_{\theta^\prime} j_{l+1}$. If $\mathcal{P}_i(\theta^\prime) = \emptyset$, then the path is $(s,i)$.
	\end{itemize}
	For each node $i$ in $T^{opt}(\theta^\prime)$, we define a \emph{weight} $w_i(T^{opt}(\theta^\prime)) = |\{j\in N^{opt}| j=i \vee i\in \mathcal{P}_{j}(\theta^\prime)\}|$, which is the total number of items allocated to $i$ and $\mathcal{C}_i(\theta^\prime)$ under the efficient allocation $\pi^{eff}(\theta^\prime)$. Let $Children(i)$ be the set of all direct children of $i$ in $T^{opt}(\theta^\prime)$.
\end{definition}

\begin{figure} 
	\centering
	\includegraphics[width=2.3in]{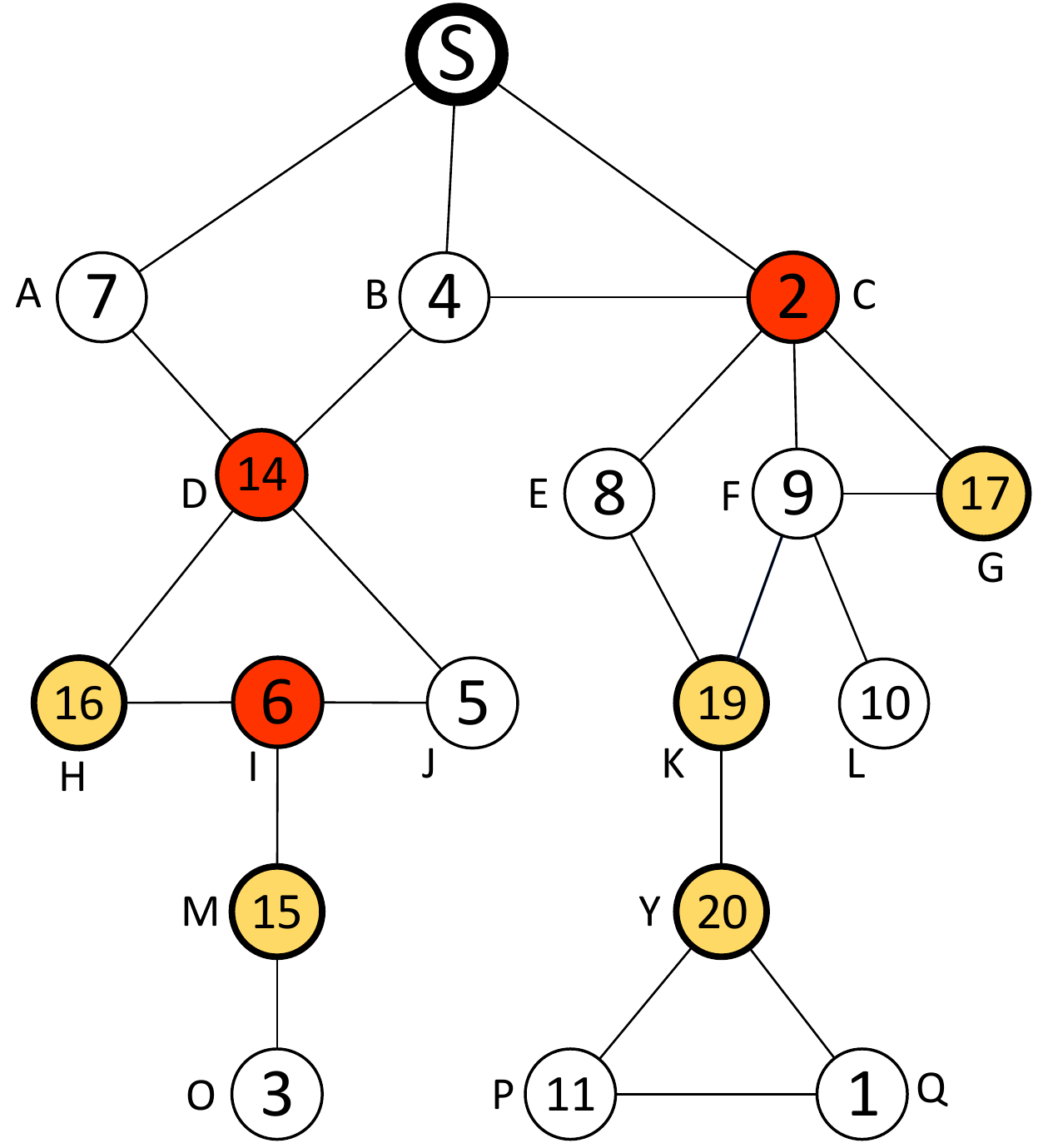}
	\caption{The optimal allocation with five items, where yellow nodes receive items and red nodes are their critical parents.}\label{eg_GIDM}
\end{figure}

\begin{figure}
	\centering
	\includegraphics[width=2.3in]{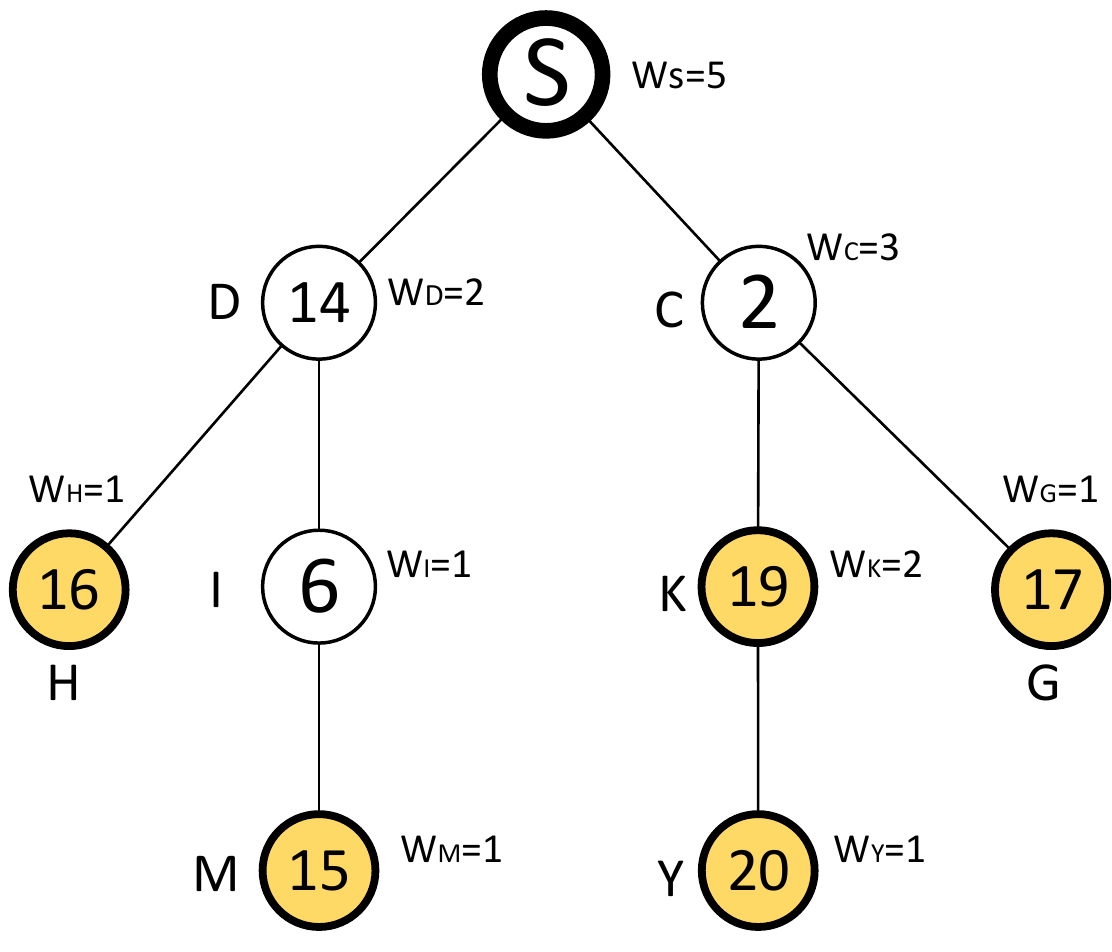}
	\caption{The optimal allocation tree $T^{opt}(\theta)$ of the allocation given in Figure~\ref{eg_GIDM}, where the weight beside each node $i$ indicates the number of items allocated to $i$ and $i$'s critical children $\mathcal{C}_i(\theta)$.}\label{eg_ct}
\end{figure}

Take the example given in Figure~\ref{eg_IDM} again, if $\mathcal{K} = 5$ and all buyers act truthfully, we have the efficient allocation as given in Figure~\ref{eg_GIDM} and its optimal allocation tree $T^{opt}(\theta)$ given in Figure~\ref{eg_ct}.

Now we are ready to define our generalization of IDM.
\begin{framed}
	\noindent\textbf{Generalized Information Diffusion Mechanism (GIDM)}\\
	\rule{\textwidth}{0.5pt}
	Given the buyers' type profile $\theta$ and their action $\theta^\prime\in \mathcal{F}(\theta)$, 
	compute the optimal allocation tree $T^{opt}(\theta^\prime)$. 
	
	
	Let $W$ be the set of buyers who receive an item in GIDM (the winners), initially $W = \emptyset$. For each $i\in W$, we define $GetFrom(i) \in N^{opt}$ to be the buyer from whom the item was taken by $i$ from the efficient allocation. 
	$GetFrom(i) = i$ indicates that $i$ takes the item from herself. 

	\begin{itemize}
		\item \textbf{Allocation}: The allocation is done with a DFS-like procedure. Let $Q$ be a last in first out (LIFO) stack, initially $Q$ is empty. The seller $s$ gives $w_i(T^{opt}(\theta^\prime))$ items to each $i\in Children(s)$ and adds all $Children(s)$  into $Q$. 
		Repeat the allocation process defined in the following until $Q$ is empty.

		\item \textbf{Payment}: For all $i\in N$, her payment is:
		\begin{equation*}
			\begin{cases}
				\mathcal{SW}_{-D_i} - (\mathcal{SW}_{-\mathcal{C}_i^{\mathcal{K}}} - v_i^\prime)	& \text{if $i\in W$,}\\
				\mathcal{SW}_{-D_i} - \mathcal{SW}_{-\mathcal{C}_i^{\mathcal{K}}} 	& \text{if $i\in \displaystyle\bigcup_{j\in W}\mathcal{P}_j(\theta^\prime) \setminus W$,}\\
				0    & \text{otherwise.}\\
			\end{cases}
		\end{equation*}
		where $\mathcal{SW}_{-\mathcal{C}_i^{\mathcal{K}}}$ is defined in the allocation section and $\mathcal{SW}_{-D_i}$ is defined by a feasible allocation $\pi$ as:
		\begin{align*}
			\text{Maximise: }  & \mathcal{SW}_{-D_i} = \sum_{j\in N_{-D_i}} \pi_j(\theta^\prime)v_j^\prime  \\
			\text{Subject to: } 	& N_{-D_i} = N \setminus D_i \\
			& D_i = \{i\} \cup \mathcal{C}_i(\theta^\prime) \\ 
				& \forall {j\in N_i^{received}}, \pi_j(\theta^\prime) = 1 \\
				& N_i^{received} = W \cap \mathcal{P}_i(\theta^\prime) \\
				&{\color{red} \forall {j \in (N^{opt}\setminus D_i ) \setminus N_i^{out}}, \pi_j(\theta^\prime) = 1}\\
				& N_i^{out} = \{j\not\in N_i^{received}| j = GetFrom(l), \forall_{l\in N_i^{received}}\}   
		\end{align*}
	\end{itemize}
\end{framed}

The intuition behind the allocation of GIDM is that if a buyer does not receive an item in the efficient allocation but her critical children receive items, then the buyer may take an item from one of her critical children, but not from any other buyer who is not her critical child. Furthermore, the buyer only takes an item if her valuation is big enough, otherwise passing it to her children gives her a higher utility.

In the definition of GIDM, $N_i^{received}$ is the set of $i$'s critical parents who have already received an item before items are passed to $i$. $N_i^{out}$ is the set of buyers who receive an item in the efficient/optimal allocation, but receive no items under GIDM as their items have been taken by their critical parents in $N_i^{received}$.

Buyer $i$ can take an item in GIDM if $i$ receives an item in the social-welfare-maximising allocation when buyers from $\mathcal{C}_i^{\mathcal{K}}$ do not participate, $i$'s critical parents who have received items, i.e. $N_i^{received}$, still receive items, and {\color{red}all buyers in $(N^{opt} \setminus \mathcal{C}_i^{\mathcal{K}} )\setminus N_i^{out}$ except for $i$ still receive items}. 

\begin{framed}
	\noindent\textbf{The Allocation of GIDM}\\
	\rule{\textwidth}{0.5pt}
	\begin{enumerate}
			\item Remove a node $i$ from $Q$, add $i$ to $W$ if $i$ receives an item in the following feasible allocation $\pi$:
			\begin{align*}
				\text{Maximise: }  & \mathcal{SW}_{-\mathcal{C}_i^{\mathcal{K}}} = \sum_{j\in N_{-\mathcal{C}_i^{\mathcal{K}}}} \pi_j(\theta^\prime)v_j^\prime  \\
				\text{Subject to: } 	& N_{-\mathcal{C}_i^{\mathcal{K}}} = N \setminus \mathcal{C}_i^{\mathcal{K}}  \\
				& \mathcal{C}_i^\mathcal{K} = \mathcal{C}_i(\theta^\prime)^\mathcal{K} \cup \mathcal{P}(\mathcal{C}_i(\theta^\prime)^\mathcal{K}) \cup \mathcal{C}(\mathcal{P}(\mathcal{C}_i(\theta^\prime)^\mathcal{K})) \\
				& \mathcal{P}(\mathcal{C}_i(\theta^\prime)^\mathcal{K}) = \displaystyle\bigcup_{j\in \mathcal{C}_i(\theta^\prime)^\mathcal{K}} \{l| l\in \mathcal{P}_j(\theta^\prime) \wedge i \succ_{\theta^\prime} l \}  \\ 
				& \mathcal{C}(\mathcal{P}(\mathcal{C}_i(\theta^\prime)^\mathcal{K})) = \bigcup_{j\in \mathcal{P}(\mathcal{C}_i(\theta^\prime)^\mathcal{K})} \mathcal{C}_j(\theta^\prime) \\
				& \forall {j\in N_i^{received}}, \pi_j(\theta^\prime) = 1 \\
				& N_i^{received} = W \cap \mathcal{P}_i(\theta^\prime) \\
				& {\color{red} \forall {j\neq i\in (N^{opt} \setminus \mathcal{C}_i^{\mathcal{K}})\setminus N_i^{out}}, \pi_j(\theta^\prime) = 1}\\
				& N_i^{out} = \{j\not\in N_i^{received}| j = GetFrom(l), \forall_{l\in N_i^{received}}\} 
			\end{align*}
			where $\mathcal{C}_i(\theta^\prime)^{\mathcal{K}}$ is the set of top $\mathcal{K}$ ranked critical children of $i$ according to their reported valuation (from high to low). If $|\mathcal{C}_i(\theta^\prime)| < \mathcal{K}$, then $\mathcal{C}_i^{\mathcal{K}} = \mathcal{C}_i(\theta^\prime)^{\mathcal{K}} = \mathcal{C}_i(\theta^\prime)$.
			\item If $i\in W$:
							\begin{itemize}
							\item if $\sum_{j \in Children(i)} w_j(T^{opt}(\theta^\prime)) = w_i(T^{opt}(\theta^\prime)) - 1$,  set $GetFrom(i) = i$,
							\item otherwise, let $k_i = w_i(T^{opt}(\theta^\prime))$, and $out$ be the buyer with the $k_i$-th largest valuation report in the subtree (of $T^{opt}(\theta^\prime)$) rooted at $i$ and $w_{out}(T^{opt}(\theta^\prime)) \neq 0$,
			for all $j \in \mathcal{P}_{out}(\theta^\prime) \cup \{out\}$ if $i \succ_{\theta^\prime} j$, set $w_j(T^{opt}(\theta^\prime)) = w_j(T^{opt}(\theta^\prime)) - 1$, and set $GetFrom(i) = out$.
							\end{itemize}			
			\item For each child $j$ of $i$, if $w_j(T^{opt}(\theta^\prime)) > 0$, give $w_j(T^{opt}(\theta^\prime))$ items to $j$ and add $j$ into $Q$.
		\end{enumerate}
\end{framed}


\begin{figure}
	\centering
	\includegraphics[width=3.5in]{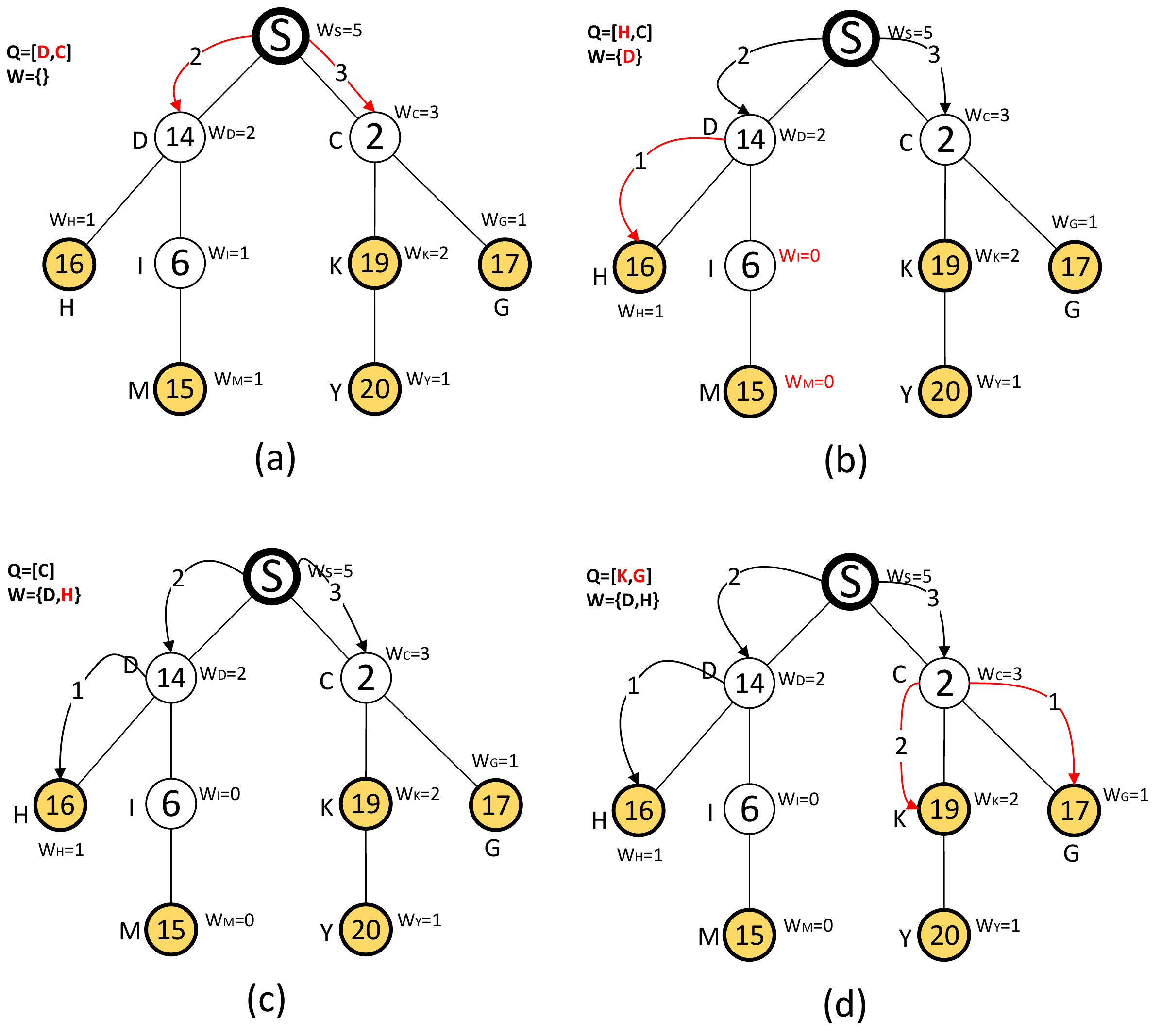}
	\caption{A running example of GIDM}\label{eg_dp}
\end{figure}

\begin{figure}
	\centering
	\includegraphics[width=3.5in]{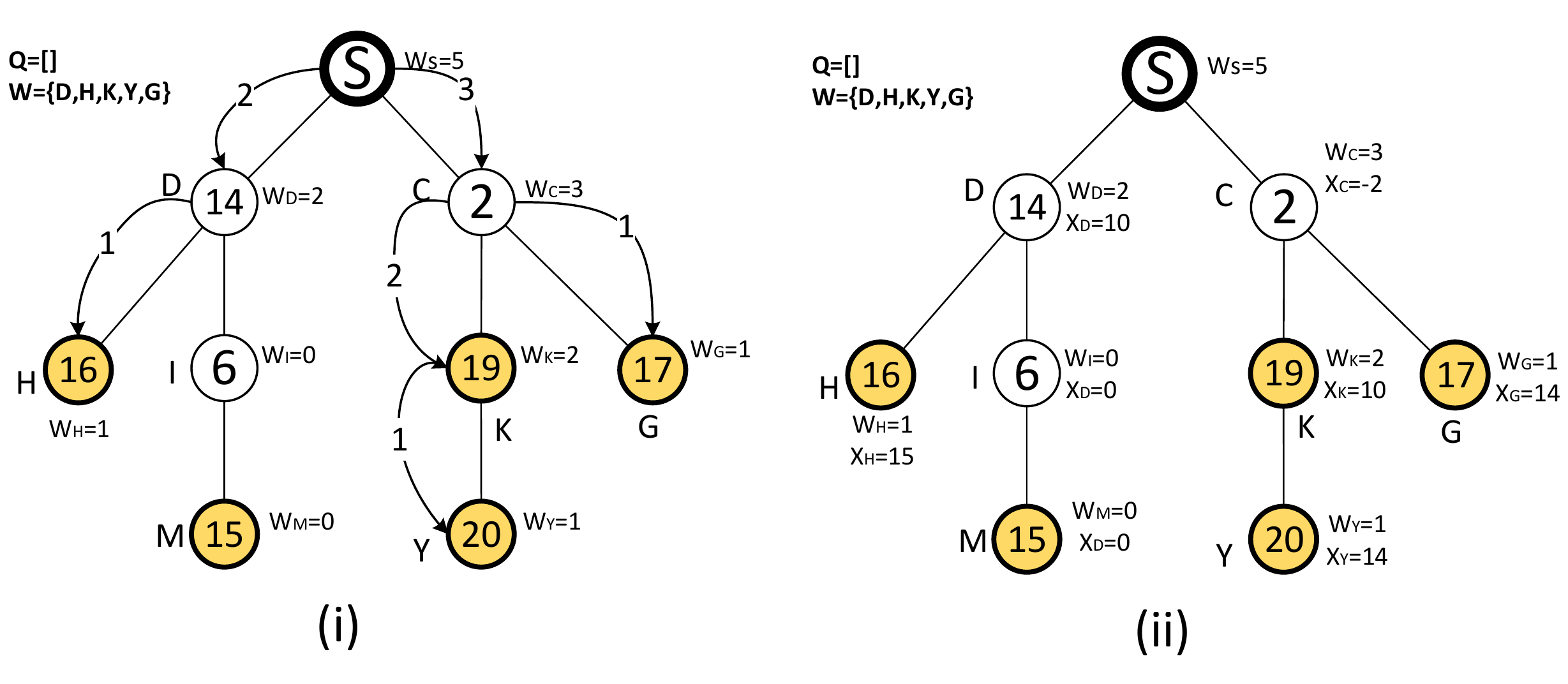}
	\caption{The outcomes of GIDM}\label{eg_res}
\end{figure}

Let $(\pi^{GIDM}, x^{GIDM})$ be the allocation policy and payment policy of the GIDM. It is easy to verify that GIDM is the same as IDM when $\mathcal{K} = 1$. Consider the social network given in Figure~\ref{eg_GIDM} with $\mathcal{K} = 5$, the GIDM runs as follows.
\begin{itemize}
	\item Firstly it computes the optimal allocation tree $T^{opt}(\theta^\prime)$ which is depicted in Figure~\ref{eg_ct}, and initializes the weight $w_i$ of each buyer in the tree according to $w_i(T^{opt}(\theta^\prime))$. This weight may be updated in the following process.
	\item Then compute the allocation of GIDM which is partially shown in Figure~\ref{eg_dp}:
	\begin{enumerate}
		\item Firstly, the seller gives $2$ items to $D$ and $3$ items to $C$, and the stack $Q$ contains $C,D$ and the winner set $W$ is empty (shown in Figure~\ref{eg_dp}$(a)$). 
		\item Then buyer $D$ is popped out of $Q$ and is identified as a winner and added in $W$ (shown in Figure~\ref{eg_dp}$(b)$). Here, $D$'s critical children $\mathcal{C}_D(\theta^\prime) = \{H,I,J,M,O\}$, and $\mathcal{C}_D^\mathcal{K} = \mathcal{C}_D(\theta^\prime)^\mathcal{K} = \{H,I,J,M,O\}$. Note that $\mathcal{C}_i(\theta^\prime)^\mathcal{K}$ and $\mathcal{C}_i^\mathcal{K}$ are normally not the same if the graph becomes complex. \\ $N_D^{received} = N_D^{out} = \emptyset$. If we remove $\mathcal{C}_D^\mathcal{K}$ from $N$, then $D$ will receive an item under $\mathcal{SW}_{-\mathcal{C}_D^{\mathcal{K}}}$, so add $D$ to $W$.
		\item In step $(2)$ of the allocation process, as $D$ did not receive an item in the optimal allocation, we have to remove one item initially allocated to $D$'s critical children by the optimal allocation. We choose the child with the lowest valuation among all $D$'s critical children who still have items, which is buyer $M$, and update the weights of $M$ and $M$'s critical parents to $D$ (as shown in Figure~\ref{eg_dp}$(b)$). Intuitively, $D$ takes the item from $M$, so we set $GetFrom(D) = M$.
		\item In step $(3)$ of the allocation process, $D$ gives $H$ one item and adds $H$ into $Q$. 
		\item Then in the next iteration, buyer $H$ is popped out and is added into $W$ (as shown in Figure~\ref{eg_dp}$(c)$). Here $\mathcal{C}_H^\mathcal{K} = \emptyset$, $N_H^{received} = \{D\}$ and $N_H^{out} = \{M\}$.
		\item {Then $C$ is popped out and is not identified as a winner. $C$ just gives $2$ items to $K$ and $1$ item to $G$ and adds $K,G$ into $Q$ (as shown in Figure~\ref{eg_dp}$(d)$). It is worth mentioning that $C$'s critical children $\mathcal{C}_C(\theta^\prime) = \{E,F,G,K,L,Y,P,Q\}$, $\mathcal{C}_C(\theta^\prime)^\mathcal{K} = \{G,K,L,Y,P\}$, $\mathcal{P}(\mathcal{C}_C(\theta^\prime)^\mathcal{K}) = \{F,K,Y\}$, and $\mathcal{C}(\mathcal{P}(\mathcal{C}_C(\theta^\prime)^\mathcal{K})) = \{L,Y,P,Q\}$. $N_C^{received} = N_C^{out} = \emptyset$. Thus, $\mathcal{C}_C^\mathcal{K}=\{F,G,K,L,Y,P,Q\}$. If we remove $\mathcal{C}_C^\mathcal{K}$ from $N$, $\mathcal{SW}_{-\mathcal{C}_C^{\mathcal{K}}}$ will allocate $\mathcal{K}$ items to $\{H,M,D,A,E\}$, and therefore $C$ cannot win.}
		\item Keep checking for the rest of the buyers of $K,G,Y$, we will end up with the allocation given in Figure~\ref{eg_res}$(i)$.
	\end{enumerate} 
	\item Lastly, compute the payments for all buyers according to the allocation. The corresponding payments of the traders in $T^{opt}(\theta^\prime)$ are given in Figure~\ref{eg_res}$(ii)$. All the other buyers who are not in $W$ receive no items and pay zero.
\end{itemize}

\section{Properties of GIDM}
\label{sect_properties}
In this section, we prove that our generalized information diffusion mechanism is individually rational, incentive compatible and improves the seller's revenue, compared with the revenue the seller can get without any advertising. Therefore, the seller is incentivized to apply our mechanism.

\begin{theorem}
	The generalized information diffusion mechanism is individually rational.
\end{theorem}
\begin{proof}
	Given the buyers' type profile $\theta$ and their action profile $\theta^\prime \in \mathcal{F}(\theta)$,
	to prove GIDM is individually rational, we need to show that for all $i\in N$, $u_i(\theta_i, \theta^\prime, (\pi^{GIDM}, x^{GIDM})) \geq 0$ for all $\theta_i^\prime = (v_i, r_i^\prime)$. 
	
	From the definition of GIDM, for any buyer $i\in N$, we have either $u_i(\theta_i, \theta^\prime) = 0$ or $u_i(\theta_i, \theta^\prime) = \mathcal{SW}_{-\mathcal{C}_i^{\mathcal{K}}} - \mathcal{SW}_{-D_i}$. According to the definitions of $\mathcal{SW}_{-\mathcal{C}_i^{\mathcal{K}}}$ and $\mathcal{SW}_{-D_i}$, we have $\mathcal{SW}_{-\mathcal{C}_i^{\mathcal{K}}} \geq \mathcal{SW}_{-D_i}$, because $N_{-\mathcal{C}_i^{\mathcal{K}}} \supset N_{-D_i}$ and the optimization under $N_{-\mathcal{C}_i^{\mathcal{K}}}$ cannot be worse than that under $N_{-D_i}$. Therefore, we have $u_i(\theta_i, \theta^\prime) \geq 0$.
\end{proof}

\begin{theorem}
	The generalized information diffusion mechanism is incentive compatible, i.e. reporting valuation truthfully and inviting all neighbours is a dominant strategy for all buyers who are aware of the sale.
\end{theorem}
\begin{proof}
	Given the buyers' type profile $\theta$ and their action profile $\theta^\prime \in \mathcal{F}(\theta)$, to prove GIDM is incentive compatible, we need to show that for each buyer $i\in N$ such that $\theta_i^\prime \neq nil$:
	\begin{itemize}
		\item fix $i$'s invitation to be $r_i^\prime$, reporting $v_i$ truthfully maximise $i$'s utility.
		\item fix $i$'s valuation report to be $v_i$, inviting all $i$'s neighbours $r_i$ maximise $i$'s utility.
	\end{itemize}
	
	We will prove the above for buyers in three different groups:
	\begin{enumerate}
		\item all buyers who receive one item, i.e. $W$.
		\item all buyers who are not in $W$, but are critical parents of $W$, i.e. $\bigcup_{i\in W}\mathcal{P}_i(\theta^\prime) \setminus W$.
		\item all buyers who are not in the first two groups.
	\end{enumerate}
	
	\textbf{Group $(1)$}: for each $i\in W$, her utility is \\
	$u_i(\theta_i, \theta^\prime) = v_i + (\mathcal{SW}_{-\mathcal{C}_i^{\mathcal{K}}} - v_i^\prime) - \mathcal{SW}_{-D_i}$.
	\begin{itemize}
		\item Fix $i$'s invitation to be $r_i^\prime$, then $\mathcal{C}_i^{\mathcal{K}}$ is fixed. 
		\begin{itemize}
			\item If $i$ reports $v_i$ truthfully, i.e. $v_i^\prime = v_i$, then $u_i(\theta_i, \theta^\prime) = \mathcal{SW}_{-\mathcal{C}_i^{\mathcal{K}}} - \mathcal{SW}_{-D_i}$. Since $\mathcal{SW}_{-\mathcal{C}_i^{\mathcal{K}}}$ is the optimal social welfare under the constraints that $i$ cannot influence, if $i$ can misreport $v_i^\prime$ to change the allocation to increase $v_i + (\mathcal{SW}_{-\mathcal{C}_i^{\mathcal{K}}} - v_i^\prime)$, then it contradicts that $\mathcal{SW}_{-\mathcal{C}_i^{\mathcal{K}}}$ is the optimal social welfare. Furthermore, $\mathcal{SW}_{-D_i}$ is independent of $i$ and we have $\mathcal{SW}_{-\mathcal{C}_i^{\mathcal{K}}} \geq \mathcal{SW}_{-D_i}$ as $N_{-\mathcal{C}_i^{\mathcal{K}}} \supset N_{-D_i}$. Therefore, $i$'s utility is maximised as soon as $i$ is still in $W$.
			
			\item Now if $i$ misreports $v_i^\prime$ such that $i$ does not receive an item and becomes a critical parent of $W$ in group $(2)$. In this case, $i$'s utility is $u_i(\theta_i, \theta^\prime) = \mathcal{SW}_{-\mathcal{C}_i^{\mathcal{K}}} - \mathcal{SW}_{-D_i}$. Since $v_i$ is not considered any more in $\mathcal{SW}_{-\mathcal{C}_i^{\mathcal{K}}}$, it means that $\mathcal{SW}_{-\mathcal{C}_i^{\mathcal{K}}}$ is at most the social welfare when $v_i$ is considered. Therefore, the utility is not better than reporting $v_i$.
			
			\item Lastly if $i$ misreports $v_i^\prime$ such that $i$ is in group $(3)$, then $u_i(\theta_i, \theta^\prime) = 0$, which is not better than reporting $v_i$ truthfully. 
		\end{itemize}
		Therefore, fixing $i$'s invitation to be $r_i^\prime$, reporting $v_i$ truthfully maximizes $i$'s utility.
		
		\item Fix $i$'s valuation report to be $v_i$, change $i$'s invitation to be any subset $r_i^\prime$ of $r_i$. For any $\theta_i^\prime = (v_i, r_i^\prime)$ and $\theta_i^{\prime\prime} = (v_i, r_i^{\prime\prime})$ such that $r_i \supseteq r_i^{\prime} \supset r_i^{\prime\prime}$, then we have $\mathcal{C}_i(\theta_i^\prime, \theta_{-i}^\prime) \supseteq \mathcal{C}_i(\theta_i^{\prime\prime}, \theta_{-i}^\prime)$, because inviting more neighbours will bring more buyers to join the sale. When we remove the buyers whose valuation is among the top $\mathcal{K}$ largest from $\mathcal{C}_i(\theta_i^\prime, \theta_{-i}^\prime)$ and $\mathcal{C}_i(\theta_i^{\prime\prime}, \theta_{-i}^\prime)$ respectively, let the corresponding $\mathcal{C}_i^\mathcal{K}$ be $\mathcal{C}_i^{\mathcal{K}, r_i^\prime}$ and $\mathcal{C}_i^{\mathcal{K}, r_i^{\prime\prime}}$. Since $\mathcal{C}_i(\theta_i^\prime, \theta_{-i}^\prime) \supseteq \mathcal{C}_i(\theta_i^{\prime\prime}, \theta_{-i}^\prime)$, we have more buyers whose action is not $nil$ in $N_{-\mathcal{C}_i^{\mathcal{K}, r_i^\prime}}$ than those in $N_{-\mathcal{C}_i^{\mathcal{K}, r_i^{\prime\prime}}}$. Therefore, we get $\mathcal{SW}_{-\mathcal{C}_i^{\mathcal{K},r_i^\prime}} \geq \mathcal{SW}_{-\mathcal{C}_i^{\mathcal{K}, r_i^{\prime\prime}}}$, i.e. $\mathcal{SW}_{-\mathcal{C}_i^{\mathcal{K},r_i^\prime}}$ is maximised when $r_i^\prime = r_i$. Thus, $u_i(\theta_i, \theta^\prime)$ is maximised when $\theta_i^\prime = \theta_i$.
	\end{itemize}
	
	\item \textbf{Group $(2)$}: for each $i\in \bigcup_{j\in W}\mathcal{P}_j(\theta^\prime) \setminus W$, her utility is \\
	$u_i(\theta_i, \theta^\prime) = \mathcal{SW}_{-\mathcal{C}_i^{\mathcal{K}}} - \mathcal{SW}_{-D_i}$.
	\begin{itemize}
		\item Fix $i$'s invitation to be $r_i^\prime$, then $\mathcal{C}_i^{\mathcal{K}}$ is fixed. 
		\begin{itemize}
			\item If $i$ misreports $v_i^\prime$, but $i$ is still in group $(2)$, then $u_i(\theta_i, \theta^\prime)$ does not change, as both $\mathcal{SW}_{-\mathcal{C}_i^{\mathcal{K}}}$ and $\mathcal{SW}_{-D_i}$ are independent of $v_i^\prime$.
			\item If $i$ misreports $v_i^\prime$ such that $i\in W$, then $u_i(\theta_i, \theta^\prime) = v_i + (\mathcal{SW}_{-\mathcal{C}_i^{\mathcal{K}}} - v_i^\prime) - \mathcal{SW}_{-D_i}$. If $ v_i + (\mathcal{SW}_{-\mathcal{C}_i^{\mathcal{K}}} - v_i^\prime)$ is greater than $\mathcal{SW}_{-\mathcal{C}_i^{\mathcal{K}}}$ when $i$ reports $v_i$, it contradicts that $\mathcal{SW}_{-\mathcal{C}_i^{\mathcal{K}}}$ is optimal. Thus, $i$'s utility is not better than reporting $v_i$.
			\item $i$ cannot misreports $v_i^\prime$ to become a member of group $(3)$.
		\end{itemize}
		
		\item Fix $i$'s valuation report to be $v_i$, inviting all neighbours $r_i$ maximises $i$'s utility. The proof is the same as the proof for group $(1)$.
	\end{itemize}
	
	\item \textbf{Group $(3)$}: for each $i$ who is not a member of group $(1)$ or group $(2)$, her utility is 
	$u_i(\theta_i, \theta^\prime) = 0$.
	\begin{itemize}
		\item Fix $i$'s invitation to be $r_i^\prime$, $i$ may misreport $v_i^\prime$ to become a member of $W$. However, in this case, $u_i(\theta_i, \theta^\prime) = v_i + (\mathcal{SW}_{-\mathcal{C}_i^{\mathcal{K}}} - v_i^\prime) - \mathcal{SW}_{-D_i}$. We know that if $i$ reports $v_i$ truthfully, we have $\mathcal{SW}_{-\mathcal{C}_i^{\mathcal{K}}} =  \mathcal{SW}_{-D_i}$. If $i$ misreports to get $v_i + (\mathcal{SW}_{-\mathcal{C}_i^{\mathcal{K}}} - v_i^\prime)$, then $v_i + (\mathcal{SW}_{-\mathcal{C}_i^{\mathcal{K}}} - v_i^\prime)$ must be less than or equal to $\mathcal{SW}_{-D_i}$. Therefore, it is not worth misreporting $v_i^\prime$.
		\item Fix $i$'s valuation report to be $v_i$, if inviting all $r_i$ does not bring $i$ to group $(1)$ or group $(2)$, then inviting less will keep $i$ in group $(3)$. Therefore, inviting all $r_i$ is the best that $i$ can do to optimise her utility.
	\end{itemize}
\end{proof}

Next we show that the seller's revenue is improved with GIDM compared with the revenue she can get with other truthful mechanisms without advertising, especially we compare it with VCG.

Without advertising, the seller can only sell the $\mathcal{K}$ items to her neighbours $r_s$. Assume that $|r_s| > \mathcal{K}$, then the revenue of applying VCG among $r_s$ is $R^{VCG} = \mathcal{K}\times v_{\mathcal{K}+1}$, where $v_{\mathcal{K}+1}$ is the $(\mathcal{K}+1)$-th largest valuation report among $r_s$. Under VCG, $i$ may improve the revenue by selling less. 
No matter how many items the seller chooses to sell under VCG, let $\mathcal{K}$ be the actually number of items that the seller is selling under both VCG and GIDM.

{Theorem~\ref{thm_rev} proves that the revenue of GIDM is not less than the revenue of VCG when the number of neighbours of the seller is more than $\mathcal{K}$. Note that when the number of the seller's neighbours is less than or equal to $\mathcal{K}$, the revenue of VCG is zero.}

\begin{theorem}
	\label{thm_rev}
	The revenue of the generalised information diffusion mechanism is greater than or equal to $\mathcal{K}\times v_{\mathcal{K}+1}$, where $v_{\mathcal{K}+1}$ is the $(\mathcal{K}+1)$-th largest valuation report among $r_s$, assume that $|r_s| > \mathcal{K}$.
\end{theorem}

Before proving the theorem, let us first show some relationship between the payment of a buyer and the payments of her direct children in the optimal allocation tree. First of all, given the buyers' action profile $\theta^\prime$, each buyers $i$'s payment under the generalised information diffusion mechanism is equal to:
	\begin{equation*}
		\begin{cases}
			(\mathcal{SW}_{-D_i} - \mathcal{V}_{N_i^{still}}) - (\mathcal{SW}_{-\mathcal{C}_i^{\mathcal{K}}}& - \mathcal{V}_{N_i^{still}} - v_i^\prime) \\
			& \text{if $i\in W$,}\\
			(\mathcal{SW}_{-D_i} - \mathcal{V}_{N_i^{still}}) - (\mathcal{SW}_{-\mathcal{C}_i^{\mathcal{K}}}& - \mathcal{V}_{N_i^{still}}) \\
			& \text{if $i\in \displaystyle\bigcup_{j\in W}\mathcal{P}_j(\theta^\prime) \setminus W$,} \\
			0    & \text{otherwise}\\
		\end{cases}
	\end{equation*}
	where 
	\begin{align*}
		& \mathcal{V}_{N_i^{still}} = \sum_{j\in N_i^{still}} v_j^\prime\\
		& N_i^{still} = (N^{opt} \setminus (D_i \cup N_i^{out})) \cup N_i^{received}\\
		& D_i = \{i\} \cup \mathcal{C}_i(\theta^\prime) \\
		& N_i^{received} = W \cap \mathcal{P}_i(\theta^\prime) \\
		&  N_i^{out} = \{j\not\in N_i^{received}| j = GetFrom(l), \forall {l\in N_i^{received}}\}  
	\end{align*}

In the payment definition of GIDM, terms $\mathcal{SW}_{-D_i}$ and $\mathcal{SW}_{-\mathcal{C}_i^{\mathcal{K}}}$ both count the valuations of all buyers in $N_i^{still}$. It is clear that $|N_i^{still}| = \mathcal{K}-k_i$. Therefore, we can remove all the valuations of $N_i^{still}$ from both $\mathcal{SW}_{-D_i}$ and $\mathcal{SW}_{-\mathcal{C}_i^{\mathcal{K}}}$ to get above form of the payments. Following this, Lemma~\ref{lem_gidm_pay2} shows the deeper relationship.

\begin{lemma}
	\label{lem_gidm_pay2}
	Given all buyers' action profile $\theta^\prime$, under the generalised information diffusion mechanism, for all $i\in \bigcup_{j\in W}\mathcal{P}_j(\theta^\prime) \setminus W$, let $k_i$ be the number of items passed to $i$, and there are $m\geq 1$ children of $i$ who have received items from $i$ denoted by $\{i_1, \cdots, i_m\}$, and let $k_{i_l}$ be the number of items $i$ gives to $i_l$, we have 
	$$\mathcal{SW}_{-\mathcal{C}_i^{\mathcal{K}}} -  \mathcal{V}_{N_i^{still}} \leq \sum_{i_l} (\mathcal{SW}_{-D_{i_l}} -  \mathcal{V}_{N_{i_l}^{still}})$$
\end{lemma}
\begin{proof}
	Given the above form of the payments, we have $\mathcal{SW}_{-\mathcal{C}_i^{\mathcal{K}}} -  \mathcal{V}_{N_i^{still}}$ is the sum of the top $k_i$ highest valuations among buyers in $N \setminus (\mathcal{C}_i^{\mathcal{K}} \cup N_i^{still})$. For each $i_l$, $\mathcal{SW}_{-D_{i_l}} - \mathcal{V}_{N_{i_l}^{still}}$ is the sum of the top $k_{i_l}$ highest valuations among buyers in $N \setminus (D_{i_l} \cup N_{i_l}^{still})$. Since $i\not\in W$, i.e. $i$ does not receive any item, we have $k_i = \sum_{i_l} k_{i_l}$, $N_i^{still} \subseteq N_{i_l}^{still}$ and $|N_{i_l}^{still} \setminus N_i^{still}| = k_i - k_{i_l}$. Since $\mathcal{K} \geq k_i \geq k_{i_l}$, we have $\mathcal{C}_i^{\mathcal{K}} \supseteq \bigcup_{i_l}D_{i_l}$ and $N_{i_l}^{still} \setminus N_i^{still} \subset \mathcal{C}_i^{\mathcal{K}}$, therefore we get $D_{i_l} \cup N_{i_l}^{still} \subseteq \mathcal{C}_i^{\mathcal{K}} \cup N_i^{still}$ and $N \setminus (D_{i_l} \cup N_{i_l}^{still}) \supseteq N\setminus (\mathcal{C}_i^{\mathcal{K}} \cup N_i^{still})$. 
	
	Since $\mathcal{SW}_{-\mathcal{C}_i^{\mathcal{K}}} -  \mathcal{V}_{N_i^{still}}$ is the sum of the top $k_i$ highest valuations in $N \setminus (\mathcal{C}_i^{\mathcal{K}} \cup N_i^{still})$, while $\mathcal{SW}_{-D_{i_l}} -  \mathcal{V}_{N_{i_l}^{still}}$ is the sum of the top $k_{i_l}$ highest valuations in a larger set $N \setminus (D_{i_l} \cup N_{i_l}^{still})$, we conclude that $\mathcal{SW}_{-\mathcal{C}_i^{\mathcal{K}}} -  \mathcal{V}_{N_i^{still}} \leq \sum_{i_l} (\mathcal{SW}_{-D_{i_l}} -  \mathcal{V}_{N_{i_l}^{still}})$.
\end{proof}

Following Lemma~\ref{lem_gidm_pay2}, we can further prove that for all buyers $i\in W$,  $$(\mathcal{SW}_{-\mathcal{C}_i^{\mathcal{K}}} -  \mathcal{V}_{N_i^{still}} - v_i^\prime) \leq \sum_{i_l} (\mathcal{SW}_{-D_{i_l}} -  \mathcal{V}_{N_{i_l}^{still}})$$
where $k_i - 1 = \sum_{i_l} k_{i_l}$ because $i$ receives an item. Furthermore, if $k_i = 1$, then $i$ does not give any item to her children and we have $\mathcal{SW}_{-\mathcal{C}_i^{\mathcal{K}}} -  \mathcal{V}_{N_i^{still}} - v_i^\prime = 0$.

\begin{proof}[Proof Theorem~\ref{thm_rev}]
	Following \ref{lem_gidm_pay2}, we conclude that for any buyers $i\in \bigcup_{j\in W}\mathcal{P}_j(\theta^\prime) \cup W$, the second term of $i$'s payment ($\mathcal{SW}_{-\mathcal{C}_i^{\mathcal{K}}} -  \mathcal{V}_{N_i^{still}}$ or $\mathcal{SW}_{-\mathcal{C}_i^{\mathcal{K}}} -  \mathcal{V}_{N_i^{still}}  - v_i^\prime$) is either $0$ or offset by the sum of all the first terms of the payments of $i$'s children 
	according to Lemma~\ref{lem_gidm_pay2}.
	
	Therefore, the sum of all buyers' payments is equal to 
	$$
	\sum_{i\in Children(s)} (\mathcal{SW}_{-D_i} -  \mathcal{V}_{N_i^{still}})  + \Delta
	$$
	where $\Delta \geq 0$. It is the sum of all the first terms of the payments of the seller's direct children in the optimal allocation tree $T^{opt}(\theta^\prime)$ plus the remaining of all the offsets.
	
	Assume $s$ has $m$ children in $T^{opt}(\theta^\prime)$, denoted by $\{s_1, \cdots, s_m\}$. The number of items passed to them are $k_{s_1}, \cdots, k_{s_m}$, where $\sum_{s_i} k_{s_i} = \mathcal{K}$. For all $s_i \in \{s_1, \cdots, s_m\}$, we have $\mathcal{SW}_{-D_{s_i}} -  \mathcal{V}_{N_{s_i}^{still}}$ is the sum of the top $k_{s_i}$ highest valuations among all buyers in $N \setminus (D_{s_i} \cup N_{s_i}^{still})$. 
	Both $D_{s_i}$ and $N_{s_i}^{still}$ may contain some buyers from $s$'s neighbours $r_s$, but these two sets cannot contain all $r_s$, i.e. $r_s \not\subseteq D_{s_i} \cup N_{s_i}^{still}$ as there are at most $\mathcal{K}$ items and $|r_s| > \mathcal{K}$, $i\in r_s \wedge i \subseteq D_{s_i} \cup N_{s_i}^{still}$ if and only if $i$ or one of $i$'s critical child has the top $\mathcal{K}$ highest valuation report among all buyers $N$. Specifically, for each $s_i$, $D_{s_i}$ contains at most one buyer from $r_s$, and $N_{s_i}^{still}$ contains at most $\mathcal{K} - k_{s_i}$ buyers from $r_s$. Therefore, the minimum of the top $k_{s_i}$ highest valuations among all buyers in $N\setminus (D_{s_i} \cup N_{s_i}^{still})$ is at least $v_{\mathcal{K}+1}$, the top $(\mathcal{K}+1)$-th largest valuation among all buyers in $r_s$. Thus, we have
	$$
	\sum_{i\in Children(s)} (\mathcal{SW}_{-D_i} -  \mathcal{V}_{N_i^{still}}) + \Delta \geq \mathcal{K}\times v_{\mathcal{K}+1}
	$$
\end{proof}

\section{Conclusions}
\label{sect_con}
We have proposed an auction mechanism that gives sale promotions to a seller to sell multiple homogeneous items via a social network. It generalises the mechanism proposed by Li et al.~\cite{li2017mechanism} for a single item setting. The mechanism is run by the seller, and she does not need to pay in advance for getting the promotions. The mechanism incentivizes all buyers who are aware of the sale to do free promotions to their neighbours, because their promotions will be rewarded if some buyers invited by them buy the items in the end. Besides the free advertising part, all buyers will also truthfully report their valuations to compete for the sale with people they have invited. Eventually, buyers who are closer to the seller will have a higher likelihood to win items than their children, because their children cannot participate in the sale without their promotions/invitations. This is the key to guarantee that all buyers are happy to invite more buyers to compete with themselves for the limited resources.

Since buyers who are closer to the sellers have the ability to control which neighbours they want to promote to, they can also control how the items are allocated. For example, when a seller sells multiple identical items to fixed number of buyers, the seller can choose to sell less with higher payments to maximise her revenue. This also applies to the buyers in our setting and they can invite more or less neighbours to control how many items are sold to their children, which may give them different rewards/utilities. On the other hand, buyers' children can also manipulate in order to satisfy their parents' needs. Therefore, it is extremely challenging to define a truthful mechanism in more complex settings. To prevent buyers' manipulations mentioned above, we have carefully chosen the allocation and payments. In particular, the definition of $\mathcal{C}_i^\mathcal{K}$ for each buyer $i$ in GIDM plays the essential role to stop their manipulations. Given $\mathcal{C}_i^\mathcal{K}$, buyer $i$'s payment does not depend on how many items her children get, therefore, she is not incentivised to control how many items her children will get.

In this paper, we assumed that each buyer only requires at most one item. We will easily lose the control if they require more than one item with different marginal valuations. Offering truthful mechanisms for general combinatorial valuation settings is highly demanded. Furthermore, we assumed that inviting neighbours in the social network does not incur a cost, e.g. posting an advertisement via facebook or twitter. However, there might be a cost to do so, so a new mechanism will be required to guarantee that buyers' promotion costs will be covered. A special social network with public diffusion/transfer costs was studied by Li et al.~\cite{li2018mechanism}, but covering diffusion costs in general networks is still open.  

\appendix
\section*{Updates}
The generalised information diffusion mechanism proposed by us in~\cite{Zhao2018} has a typo in one constraint. What we needed is a weaker constraint, which is implied by the misspecified constraint though, because the additional limitation given by the misspecified one may affect the properties of the mechanism. We correct the typo in this update, which is highlighted. Note that, all the results and the proofs presented in \citep{Zhao2018} are not affected by this correction.

Under the description of the generalised information diffusion mechanism:
\begin{enumerate}
\item Under the payment definition, in the definition of $\mathcal{SW}_{-D_i}$, change the constraint $$\forall {j \in N_i^{out}}, \pi_j(\theta^\prime) = 0$$ to 
{\color{red}$$\forall {j \in (N^{opt}\setminus D_i) \setminus N_i^{out}}, \pi_j(\theta^\prime) = 1$$}

\item Under the allocation definition, in the definition of $\mathcal{SW}_{-\mathcal{C}_i^{\mathcal{K}}}$, change the constraint 
$$\forall {j\neq i\in N_i^{out}}, \pi_j(\theta^\prime) = 0$$
to
{\color{red}$$\forall {j\neq i\in (N^{opt} \setminus \mathcal{C}_i^{\mathcal{K}} )\setminus N_i^{out}}, \pi_j(\theta^\prime) = 1$$}

\item In the third paragraph under the description of the generalised information diffusion mechanism, change the sentence "\textbf{all buyers in $N_i^{out}$ except for $i$ do not receive items}" to "{\color{red}all buyers in $(N^{opt} \setminus \mathcal{C}_i^{\mathcal{K}} )\setminus N_i^{out}$ except for $i$ still receive items}".
\end{enumerate}

The differences between the original constraints and the corrections are that:
\begin{itemize}
\item In the definition of $\mathcal{SW}_{-D_i}$, the original constraint $$\forall {j \in N_i^{out}}, \pi_j(\theta^\prime) = 0$$ implies the correction $$\forall {j \in (N^{opt}\setminus D_i) \setminus N_i^{out}}, \pi_j(\theta^\prime) = 1$$
The original constraint means that all buyers in $N_i^{out}$ do not receive items, while the correction indicates that all buyers in $(N^{opt}\setminus D_i) \setminus N_i^{out}$ still receive items, i.e. buyers in $N_i^{out}$ are allowed to receive items if their valuations are large enough under $\mathcal{SW}_{-D_i}$, but they cannot take items from the initial winners in $(N^{opt}\setminus D_i) \setminus N_i^{out}$ (even if the buyers' valuations in $N_i^{out}$ are greater than the buyers' valuations in $(N^{opt}\setminus D_i) \setminus N_i^{out}$). Under the original constraint, in some rare case, a buyer may over report her valuation to receive one item but pays something less than her true valuation, which affects the IC property. 
\item a similar logic applies to the second correction.
\end{itemize}


\bibliographystyle{ACM-Reference-Format}  
\newpage


\end{document}